\documentclass[letterpaper]{scrartcl}
\usepackage[total={16.2cm, 23.5cm}]{geometry}
\usepackage{amsmath,amsthm,amssymb}
\usepackage{microtype}
\usepackage[T1]{fontenc}
\usepackage{lmodern}
\usepackage{txfonts}
\usepackage{tikz}
\usetikzlibrary{patterns,decorations.pathreplacing, arrows.meta}
\usepackage{hyperref}
\hypersetup{
	pdfencoding=auto, psdextra,
	colorlinks=true,citecolor=green!30!black,linkcolor=red!40!black,urlcolor=blue!60!black
}

\usepackage[sort&compress,nameinlink]{cleveref}
\usepackage{nicefrac}
\usepackage[numbers,sort]{natbib}
\usepackage{comment}
\usepackage{paralist}
\usepackage[shortlabels]{enumitem}
\usepackage{xcolor}
\usepackage{colortbl}
\usepackage{subcaption}
\usepackage{booktabs}
\usepackage{listings}
\usepackage{algorithm2e}
\definecolor{winered}{rgb}{0.5,0.1,0.1}
\usepackage{fontawesome}
\usepackage{tabularx}

\newcolumntype{Y}{>{\centering\arraybackslash}X}
\newcolumntype{Z}{>{\raggedleft\arraybackslash}X}
\newcolumntype{a}{>{\columncolor{black!05}}Y}

\renewcommand*{\le}{\leqslant}
\renewcommand*{\leq}{\leqslant}
\renewcommand*{\ge}{\geqslant}
\renewcommand*{\geq}{\geqslant}
\renewcommand{\epsilon}{\varepsilon}

\crefname{table}{Table}{Tables}
\Crefname{table}{Table}{Tables}
\crefname{figure}{Figure}{Figures}
\crefname{theorem}{Theorem}{Theorems}
\crefname{definition}{Definition}{Definitions}
\crefname{corollary}{Corollary}{Corollaries}
\crefname{observation}{Observation}{Observations}
\crefname{question}{Question}{Question}
\crefname{lemma}{Lemma}{Lemmas}
\crefname{example}{Example}{Examples}
\crefname{reduction}{Reduction}{Reductions}
\crefname{construction}{Construction}{Constructions}
\crefname{subsection}{Section}{Sections}
\crefname{section}{Section}{Sections}
\crefname{proposition}{Proposition}{Propositions}
\crefname{algorithm}{Algorithm}{Algorithms}
\crefname{algocf}{Algorithm}{Algorithms}
\Crefname{equation}{Inequality}{Inequalities}
\crefname{lstlisting}{listing}{listings}

\newcommand{\myemph}[1]{{\color{winered}\emph{#1}}}
\newcommand{\naturals}{{{\mathbb{N}}}}

\newcommand{\reals}{{{\mathbb{R}}}}

\newcommand{\cost}{{{\mathrm{cost}}}}

\renewcommand{\part}{{{\mathrm{part}}}}

\theoremstyle{definition}
\newtheorem{definition}{Definition}

\newtheorem{example}{Example}
\newtheorem{remark}{Remark}
\theoremstyle{plain}
\newtheorem{theorem}{Theorem}

\newtheorem{corollary}[theorem]{Corollary}
\newtheorem{proposition}[theorem]{Proposition}

\begin{document}
	\title{Robust Proportionality Axioms \\ for Additive Utilities}
	\author{
	Piotr Skowron\\
  	University of Warsaw\\
  	{\small \url{p.skowron@mimuw.edu.pl}}}
	\date{}

	\maketitle
	
	\begin{abstract}
	\end{abstract}
	
	\lstset{
		keywords={input, output, for, while, if, else, return, break},
		comment=[l]{//},
		frame=single,
		mathescape=true,
		float,
		captionpos=b,
		numbers=left,
		breaklines=true,
	}
	
\begin{abstract}
We study the model of participatory budgeting, in which a subset of candidates is selected subject to a budget constraint. Voters express their preferences by assigning a numerical utility to each candidate. Our focus is on formal analysis of proportionality, a central notion of fairness in multi-winner elections. Intuitively, proportionality means that each group of voters should influence the election outcome roughly in proportion to its size. Recent ``plus'' axioms of proportionality are particularly appealing notions: they are verifiable in polynomial time, and they provide strong representation guarantees without requiring groups of voters to have cohesive preferences. However, their existing formulations are restricted to approval-based preferences. We generalize extended justified representation plus (EJR+), and Core plus (Core+) to general additive utilities. Our approach builds on a payment-based interpretation of proportionality: candidates' costs are shared among their supporters, with the aim of balancing voters' total payments.
\end{abstract}
	
\section{Introduction}

We consider the participatory budgeting model~\cite{participatoryBudgeting, rey-mal:survey-on-pb, pet-pie-sko:c:participatory-budgeting-cardinal}, in which there is a set of candidates, each associated with a cost, and the goal is to select a subset of these candidates, whose total cost does not exceed a given budget. Candidates are selected based on the preferences of a set of voters. We assume that each voter expresses her preferences by \myemph{assigning a utility to each candidate}: a non-negative number measuring how much the voter values that candidate. 

The goal is to select an outcome that proportionally reflects voters' preferences. Intuitively, each group of voters should be able to determine how a share of the budget roughly proportional to its size is spent. In participatory budgeting, proportionality ensures that the decision about public spending takes into account the diversity of opinions within the population, and that the decision is not dictated solely by one (typically the largest) group of voters. If all candidates have equal costs, then our framework collapses to the model of committee elections~\cite{FSST-trends, lac-sko:multiwinner-book}. There, proportionality is essential to ensure that minorities of voters are represented in the elected body. The participatory budgeting model also applies beyond the political domain, and there proportionality serves other critically important purposes.\footnote{For example, proportional  rules are used in blockchain consensus protocols to guard against disproportionate influence by coordinated groups of users~\cite{cevallos2020validator, burdges2020overview}. Such rules have also been used to improve the performance of genetic algorithms~\cite{fal-saw-sch-smo:c:multiwinner-genetic-algorithms}.}

In recent years, researchers have formalized the intuition behind proportionality in numerous ways, giving rise to a rich landscape of proportionality axioms~\cite{justifiedRepresentation, pjr17, rey-mal:survey-on-pb, pet-sko:laminar, BP-ejrplus, don-pet:axioms-for-axioms}. Among these, the ``plus'' axioms stand out. \citet{BP-ejrplus} introduced extended justified representation plus (EJR+). \citet{teh:fjr-plus} and \citet{bec-gre-pet:core} subsequently proposed the stronger axioms of full justified representation plus (FJR+) and Core plus (Core+). These strengthen the axioms of EJR~\cite{justifiedRepresentation}, FJR~\cite{pet-pie-sko:c:participatory-budgeting-cardinal}, and Core\cite{justifiedRepresentation}, respectively, and offer several key advantages over their original counterparts:
\begin{inparaenum}[(1)]
\item it can be verified in polynomial time whether a given outcome satisfies these axioms;
\item they provide strong representation guarantees for groups of voters, without requiring these groups to have cohesive preferences; and
\item they are among the strongest proportionality axioms known to be satisfiable.
\end{inparaenum}

These axioms were originally proposed for approval-based committee elections. With approval ballots, each voter simply indicates the subset of candidates she \myemph{approves}. Such ballots can be interpreted within our additive-utility model in two standard ways~\cite{fal-fli-pet-pie-sko-sto-szu-tal:pb-experiments}. Under \myemph{approval utilities}, each voter assigns utility one to every approved candidate and zero to all others. Under \myemph{cost utilities}, each voter assigns to every approved candidate a utility equal to its cost and zero to all others. The latter approach reflects the assumption that voters prefer larger, and more costly candidates, among the ones they approve of. The axiom of EJR+ has been defined only for approval and cost utilities. Similarly, Core+ has so far been formulated only for approval utilities.
\myemph{In this paper, we generalize EJR+ and Core+ to general additive utilities.} In related work, \citet{neo-teh:fjr-additive} extended FJR+ to participatory budgeting with additive utilities, obtaining a polynomial-time verifiable axiom that implies FJR up to one.

Our approach is based on the idea of payments: the cost of each candidate should be shared among the voters who support it. From this perspective, proportionality expresses the intuitive requirement that voters' total payments should be as balanced as possible. Indeed, our first step is to reformulate EJR+ in terms of payments. Then, the payment-based formulation naturally extends to additive utilities by imposing that each candidate's cost should be shared among its supporters in proportion to their utilities for that candidate. This reflects the idea of minimizing the payment-per-utility ratio, which also appears elsewhere in the literature, notably in the design of the Method of Equal Shares~\cite{pet-pie-sko:c:participatory-budgeting-cardinal}. The same approach yields analogous extensions of Core+ to additive utilities.

\section{Preliminaries}

A participatory budgeting election, or in short an \myemph{election}, is a triple $E = (C, V, b)$, where $C = \{c_1, c_2, \ldots, c_m\}$ is a nonempty set of $m$ \myemph{candidates}, $V = \{v_1, v_2, \ldots, v_n\}$ is a nonempty set of $n$ \myemph{voters}, and $b$ is a \myemph{budget} value, $b > 0$. For each candidate $c \in C$, let $\cost(c)$ denote its cost, $\cost(c) > 0$, and let $\cost(W) = \sum_{c \in W} \cost(c)$ for each $W \subseteq C$. We say that a subset of candidates $W \subseteq C$ is \myemph{feasible} if $\cost(W) \leq b$. Each voter $v_i \in V$ assigns to each candidate $c \in C$ a nonnegative \myemph{utility}, denoted as $u_i(c)$. We extend this notation to sets; for a voter $v_i\in V$, a group $S\subseteq V$, and a set
$W\subseteq C$, we write
\begin{align*}
    u_i(W)=\sum_{c\in W}u_i(c),
    \qquad
    u_S(c)=\sum_{v_i\in S}u_i(c),
    \qquad
    u_S(W)=\sum_{v_i\in S}u_i(W).
\end{align*}
We say that a subset of candidates $W$ is \myemph{exhaustive} if for each $c \notin W$ we have that $\cost(W) + \cost(c) > b$.

\subsection{Approval Preferences and Committee Elections}

There are two natural subclasses of voters' preferences, where the voters simply indicate which candidates they approve. We say that the voters have \myemph{approval utilities} if for each voter $v_i \in V$ and candidate $c \in C$ either $u_i(c) = 0$ or $u_i(c) = 1$. We say that the voters have \myemph{cost utilities} if for each $v_i \in V$ and $c \in C$ either $u_i(c) = 0$ or $u_i(c) = \cost(c)$. For a discussion of similarities and differences between the two utility models we refer to the paper by \citet{fal-fli-pet-pie-sko-sto-szu-tal:pb-experiments}. 

By $A_i$ we denote the set of candidates approved by voter $v_i$, that is:
\begin{align*}
A_i = \{c \in C \colon u_i(c) > 0\} \text{.}
\end{align*}
Analogously, we define the set of voters who approve candidate $c$ as $V_c$:
\begin{align*}
V_c = \{v_i \in V \colon c \in A_i\} \text{.}
\end{align*}
		
We call an election an instance of \myemph{committee elections} if every candidate has cost equal to one. In other words, in committee elections the goal is to select a subset of at most $b$ candidates.

\subsection{Selection Rules}

A \myemph{selection rule} is a function that for each election $E$ returns a nonempty set of feasible outcomes.\footnote{Typically, we are interested in a single outcome, though we allow for ties.} For example, the \myemph{utilitarian rule} selects feasible subsets $W$ that maximize the total utility $\sum_{v_i \in V}u_i(W)$. Another example is the \myemph{greedy utilitarian rule}, which is an algorithm that starts with an empty outcome $W = \emptyset$. In each iteration, it selects the candidate $c \in C \setminus W$ such that $W \cup \{c\}$ is feasible, and $\sum_{v_i \in V}u_i(c) / \cost(c)$ is maximized, and then adds $c$ to $W$. The algorithm stops when $W$ becomes exhaustive.

\subsubsection*{Proportional Approval Voting}

The \myemph{Proportional Approval Voting} (PAV) rule~\cite{justifiedRepresentation, lac-sko:multiwinner-book} is defined for committee elections with approval utilities. For each subset of candidates $W\subseteq C$, its \myemph{PAV score} is defined as
\begin{align}\label{eq:pav-score}
    \operatorname{sc}_{\mathrm{PAV}}(W) = \sum_{v_i \in V}\sum_{j=1}^{|A_i\cap W|} \frac{1}{j},
\end{align}
where an empty sum is interpreted as zero. PAV selects the feasible committees that maximize this score. 

\begin{example}[Sample PAV Execution]
Consider a committee election with five voters, three candidates $C=\{c_1, c_2, c_3\}$, and a budget of $b=2$. The approval sets are
\begin{align*}
    A_1=A_2=A_3=\{c_1,c_2\},
    \qquad
    A_4=A_5=\{c_3\}.
\end{align*}
The scores of all size-two committees are as follows:
\begin{center}
    \renewcommand{\arraystretch}{1.3}
    \begin{tabular}{lcc}
        \toprule
        Committee $W$
            & $\operatorname{sc}_{\mathrm{PAV}}(W)$
            & $\sum_{v_i\in V}u_i(W)$ \\
        \midrule
        $\{c_1,c_2\}$
            & $3\cdot(1+\frac12)=\frac92$
            & $6$ \\
        $\{c_1,c_3\}$
            & $3\cdot 1+2\cdot 1=5$
            & $5$ \\
        $\{c_2,c_3\}$
            & $3\cdot 1+2\cdot 1=5$
            & $5$ \\
        \bottomrule
    \end{tabular}
\end{center}

Consequently, PAV returns the two committees $\{c_1,c_3\}$ and $\{c_2,c_3\}$. Each of these committees gives every voter exactly one approved representative. In contrast, the utilitarian rule selects $\{c_1,c_2\}$, which maximizes total utility but leaves $v_4$ and $v_5$ without any approved representative. \hfill $\lrcorner$
\end{example}

\subsubsection*{The Method of Equal Shares}

The \myemph{Method of Equal Shares} (MES)~\cite{pet-sko:laminar,pet-pie-sko:c:participatory-budgeting-cardinal} is a sequential rule based on the idea of payments. It starts with an empty outcome $W=\emptyset$ and gives each voter $v_i$ an equal share of the budget, $b_i = \nicefrac{b}{n}$. Throughout the execution, $b_i$ denotes the remaining budget of voter $v_i$.

We say that an unelected candidate $c\in C\setminus W$ is \myemph{$\rho$-affordable}, for $\rho \geq 0$, if
\begin{align}\label{eq:mes-affordability}
    \sum_{v_i \in V}\min\left(b_i, \rho u_i(c)\right) = \cost(c).
\end{align}
Intuitively, the above equality says that the cost of the candidate can be spread among the voters so that the following holds: each voter $v_i$ pays $\rho \cdot u_i(c)$, unless this exceeds her remaining budget, in which case she spends all her remaining money. Thus, voters have possibly equal payment-per-utility ratios, subject to voters' budget constraints.

In each iteration, MES selects a candidate that is $\rho$-affordable for the smallest possible value of $\rho$. It adds this candidate $c$ to $W$ and updates each voter's budget as
follows:
\begin{align*}
    b_i \leftarrow b_i - \min (b_i,\rho u_i(c)).
\end{align*}
If no unelected candidate is $\rho$-affordable for any $\rho\geq 0$, the algorithm stops and returns $W$.

The outcome of Equal Shares does not need to be exhaustive, as the following example illustrates.

\begin{figure}[t]
    \centering
    \begin{tikzpicture}[x=0.85cm,y=0.85cm,font=\small]

        \begin{scope}
            \node at (1.325,3.2) {Initially};
            \foreach \x/\i in {0/1,1/2,2/3} {
                \filldraw[fill=black!15,draw=black]
                    (\x,0) rectangle ++(0.65,2);
                \node[below] at ({\x+0.325},0) {$v_{\i}$};
                \node[above] at ({\x+0.325},2.1) {$2$};
            }
        \end{scope}

        \begin{scope}[shift={(5,0)}]
            \node at (1.325,3.2) {After selecting $c_2$};
            \foreach \x/\h/\i in {0/2/1,1/1/2,2/1/3} {
                \draw[dashed,black!50]
                    (\x,0) rectangle ++(0.65,2);
                \filldraw[fill=black!15,draw=black]
                    (\x,0) rectangle ++(0.65,\h);
                \node[below] at ({\x+0.325},0) {$v_{\i}$};
                \node[above] at ({\x+0.325},2.1) {$\h$};
            }
        \end{scope}

        \begin{scope}[shift={(10,0)}]
            \node at (1.325,3.2) {After selecting $c_1$};
            \foreach \x/\h/\i/\lab in {
                0/0.666667/1/{\frac{2}{3}},
                1/0.333333/2/{\frac{1}{3}},
                2/1/3/{1}
            } {
                \draw[dashed,black!50]
                    (\x,0) rectangle ++(0.65,2);
                \filldraw[fill=black!15,draw=black]
                    (\x,0) rectangle ++(0.65,\h);
                \node[below] at ({\x+0.325},0) {$v_{\i}$};
                \node[above] at ({\x+0.325},2.1) {$\lab$};
            }
        \end{scope}

        \draw[->,thick] (3.05,1.15) -- (4.55,1.15)
            node[midway,above] {$\rho=\frac{1}{2}$};
        \draw[->,thick] (8.05,1.15) -- (9.55,1.15)
            node[midway,above] {$\rho=\frac{2}{3}$};
    \end{tikzpicture}
    \caption{An execution of the Method of Equal Shares.
        Filled bars show remaining budgets; dashed outlines mark
        the initial budget of $2$. The numbers above the bars give
        the remaining balances. MES selects $c_2$, then $c_1$, and
        stops with $2$ budget units unspent.}
    \label{fig:mes-example}
\end{figure}
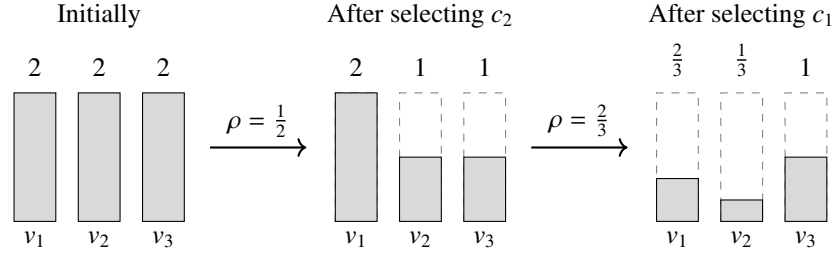

\begin{example}[Sample MES Execution]
Consider an election with three voters, four candidates, and a budget of $b=6$. The costs and utilities are given below:
\begin{center}
    \renewcommand{\arraystretch}{1.15}
    \begin{tabular}{lcccc}
        \toprule
        Candidate & $c_1$ & $c_2$ & $c_3$ & $c_4$ \\
        \midrule
        Cost  & $2$ & $2$ & $2$ & $3$ \\
        $u_1$ & $2$ & $0$ & $1$ & $0$ \\
        $u_2$ & $1$ & $2$ & $0$ & $1$ \\
        $u_3$ & $0$ & $2$ & $1$ & $1$ \\
        \bottomrule
    \end{tabular}
\end{center}

Initially, each voter has a budget of $2$. The smallest values of $\rho$ for which $c_1,c_2,c_3,c_4$ are $\rho$-affordable are, respectively, $\nicefrac{2}{3}$, $\nicefrac{1}{2}$, $1$, and $\nicefrac{3}{2}$. Hence, MES first selects $c_2$. Voters $v_2$ and $v_3$ each pay $1$, while $v_1$ pays nothing. Their remaining budgets are therefore $(b_1,b_2,b_3)=(2,1,1)$.

In the next iteration, $c_1$ is still $\nicefrac{2}{3}$-affordable, $c_3$ is $1$-affordable, and $c_4$ is no longer affordable: its supporters have only the following amount of remaining money: $b_2+b_3=2<3=\cost(c_4)$. Thus, MES selects $c_1$. Voters $v_1$ and $v_2$ pay $\nicefrac{4}{3}$ and $\nicefrac{2}{3}$, respectively.
The remaining budgets become $(b_1,b_2,b_3)=\left(\nicefrac{2}{3}, \nicefrac{1}{3}, 1\right)$. The supporters of $c_3$ now have only $b_1 + b_3 = \nicefrac{5}{3} < 2$ remaining money, while those of $c_4$ have only $b_2 + b_3 = \nicefrac{4}{3} < 3$. Consequently, MES stops and returns $W=\{c_1, c_2\}$. \Cref{fig:mes-example} illustrates the changes in voters' budgets.

Although $\cost(W) = 4$ and the total remaining budget is $2$, candidate $c_3$ cannot be selected by MES: its supporters cannot cover its cost. Thus, the returned outcome is not exhaustive. \hfill $\lrcorner$
\end{example}

\section{The Axiom of EJR+ for Additive Utilities}
		
We start by considering the extensions of the axiom of Extended Justified Representation (EJR)~\cite{justifiedRepresentation,pet-pie-sko:c:participatory-budgeting-cardinal}.
While EJR is one of the most fundamental properties considered in the context of proportionality, it has two major drawbacks. First, verifying whether a given outcome satisfies EJR is coNP-hard~\cite{justifiedRepresentation}. This makes the axiom impractical for empirical evaluation of election rules: it is infeasible to check how often certain rules violate EJR in data. This limitation is critical when we are interested in rules that do not satisfy EJR in the worst case but still adhere to the axiom for the vast majority of election instances. Second, the axiom often fails to eliminate many disproportional outcomes, as large cohesive groups are rare in practice~\cite{bre-fal-kac-nie2019:experimental_ejr}.
		
The recently proposed axiom of EJR+~\cite{BP-ejrplus} addresses these two issues: it is verifiable in polynomial time and provides strictly stronger guarantees without referring to the concept of cohesiveness. Yet, the axiom has been proposed only for approval and cost utilities, and it remained an open problem how to extend it to additive valuations. In this section, we positively resolve this issue. Let us start by first recalling the definition of EJR+ for approval and cost utilities, and then by introducing a slightly weaker variant of the axiom of EJR+ to build a better intuition.
		
\begin{definition}[EJR+ for Approval Committee Elections~\cite{BP-ejrplus}]\label{def:standard-ejr-plus}
	Given a committee election $E = (C, V, b)$ with approval utilities, we say that an outcome $W \subseteq C$ satisfies EJR+ if for each $\ell \in \naturals$, each group of voters $S \subseteq V$ with $|S| \geq \ell \cdot \nicefrac{n}{b}$, and each candidate $c \in \cap_{v_i \in S}A_i \setminus W$ there is a voter $v_i \in S$ with $|A_i \cap W| \geq \ell$. \hfill $\lrcorner$
\end{definition}
		
\begin{definition}[EJR+ up-to-one for Cost Utilities~\cite{BP-ejrplus}]
	Given an election $E = (C, V, b)$ with cost utilities, an outcome $W \subseteq C$ satisfies EJR+ up-to-one if for each group of voters $S \subseteq V$ and each candidate $c \in \cap_{v_i \in S}A_i \setminus W$ there is a voter $v_i \in S$ such that:
	\begin{align*}
		\cost(A_i \cap W) + \cost(c) > |S| \cdot \nicefrac{b}{n} \text{.}\tag*{$\lrcorner$}
	\end{align*}
\end{definition}
		
\subsection{Weak Extended Justified Representation Plus (Weak EJR+)}
The following calculations provide intuition for our formulation of the axiom; they do not yet constitute formal proofs. The underlying idea is that the cost of each selected project should be shared among its supporters in proportion to their utilities for that project, and that voters' total spending should be as equal as possible. This idea is reflected, in particular, in the formulation of the Method of Equal Shares~\cite{pet-sko:laminar, pet-pie-sko:c:participatory-budgeting-cardinal}, and here we will build on it, in order to recover the formulation of the particular EJR+ axioms.
		
Consider an unelected candidate $c \in C$ and a group of voters $S$. If we spread the cost of $c$ proportionally among the voters from $S$, then each voter $v_i \in S$ would pay:
\begin{align*}
	p_i(c) = \cost(c) \cdot \frac{u_i(c)}{\sum_{v_j \in S} u_j(c)} \text{.}
\end{align*}
Then, the ratio of the price per gained utility would be equal among the voters from $S$. Let us denote this value as $\rho(c)$:
\begin{align*}
	\rho(c) = \frac{p_i(c)}{u_i(c)} = \frac{\cost(c)}{\sum_{v_j \in S} u_j(c)} \text{.}
\end{align*}
Assuming that the voters made the purchases in $W$ as efficiently as if they bought $c$, each such voter $i \in S$ would need to spend the following amount of money on $W$:
\begin{align*}
	m_i \leq u_i(W) \cdot \rho(c) =   u_i(W) \cdot \frac{\cost(c)}{\sum_{v_j \in S} u_j(c)} \text{.}
\end{align*}
If they were about to buy additionally $c$, then each of them would spend $m_i + p_i(c)$. Since $c$ has not been eventually selected, such a purchase should not be possible. Hence, for at least one voter $i \in S$ it must hold that:
\begin{align*}
	m_i + p_i(c) > \frac{b}{n} \text{.}
\end{align*}
This leads us to the following formulation of the axiom of EJR+ for additive utilities.
		
\begin{definition}[Weak EJR+ for Additive Utilities]
	An outcome $W \subseteq C$ satisfies weak EJR+ if for each unelected candidate $c \notin W$ and each group $S \subseteq V$ with $u_S(c) > 0$ there exists a voter $v_i \in S$ such that:
	\begin{align}\label{eq:weak-ejr}
		\frac{\cost(c)}{\sum_{v_j \in S} u_j(c)} \cdot u_i(W) \geq \frac{b}{n} \text{.}
	\end{align}
	An outcome $W$ satisfies weak EJR+ up-to-one if for all $c \notin W, S \subseteq V$ with $u_S(c) > 0$, there is $v_i \in S$ such that
	\begin{align}\label{eq:weak-ejr-up-to-one}
		\frac{\cost(c)}{\sum_{v_j \in S} u_j(c)} \cdot (u_i(c) + u_i(W)) > \frac{b}{n} \text{.}
	\end{align}
	\hfill $\lrcorner$
\end{definition}
		
\begin{theorem}
	Weak EJR+ up-to-one is equivalent to EJR+ for committee elections with approval utilities and to EJR+ up-to-one for cost utilities.
\end{theorem}
\begin{proof}
	Let us first consider the setting of committee elections with approval utilities. Then Inequality \eqref{eq:weak-ejr-up-to-one} boils down to:
	\begin{align*}
		\frac{1}{|V_c \cap S|} \cdot (|W \cap A_i| + 1) > \frac{b}{n} \text{.}
	\end{align*}
	Note that it is sufficient to consider only $S \subseteq V_c$. Indeed consider a set $S$ with $S \setminus V_c \neq \emptyset$. If the condition is satisfied for $S \cap V_c$, then it would also be satisfied for $S$ (the denominator does not change, and the inequality has to be satisfied for at least one voter from a larger set). If it fails for $S \cap V_c$, then the axiom is violated (since the definition requires it must be satisfied for all sets $S$). Thus, the definition of weak EJR+ for additive utilities says that for each group of voters $S \subseteq V$ and for each candidate $c \in \cap_{v_i \in S}A_i \setminus W$ there exists a voter $v_i \in S$ such that:
	\begin{align*}
		|W \cap A_i| + 1 > \frac{b}{n} \cdot |S| \text{.}
	\end{align*}
	Let $\ell$ be the largest value such that $|S| \geq \ell \cdot \nicefrac{n}{b}$. Then the condition above is equivalent to:
	\begin{align*}
		|W \cap A_i| \geq \left\lfloor \frac{b}{n} \cdot |S| \right\rfloor = \left\lfloor \frac{b}{n} \cdot \ell \cdot \nicefrac{n}{b} \right\rfloor = \ell\text{.}
	\end{align*}
	This shows the equivalence with EJR+ for approval committee elections.
			
	Let us now move to the case of cost utilities. Now, Inequality \eqref{eq:weak-ejr-up-to-one} boils down to:
	\begin{align*}
		\frac{b}{n} &<  \frac{\cost(c)}{\sum_{v_j \in S} u_j(c)} \cdot (u_i(c) + u_i(W))  \\
			& = \frac{\cost(c)}{\sum_{v_j \in S \cap V_c} \cost(c)} \cdot (u_i(c) + u_i(W)) \\
			& = \frac{1}{|S \cap V_c|} \cdot (u_i(c) + u_i(W)) \text{.}
	\end{align*}
	Again, we can consider only $S \subseteq V_c$. Thus, the axiom of weak EJR+ up-to-one for additive utilities says that for each group of voters $S$ and each candidate $c \in \cap_{v_i \in S}A_i \setminus W$ there is a voter $v_i \in S$ such that:
	\begin{align*}
		\frac{b}{n} \cdot |S| <  u_i(c) + u_i(W) = \cost(c) + \cost(W \cap A_i) \text{.}
	\end{align*}
	This shows the equivalence with the axiom of EJR+ up-to-one for cost utilities.
\end{proof}
		
\subsection{Strong Extended Justified Representation Plus (Strong EJR+)}
		
We will now propose a strengthening of the axiom. Again, consider an unelected candidate $c \in C$. We observe that in our intuitive reasoning the payment of a voter $i \in S$ for an elected candidate $c'$ can be bounded using a stronger inequality. So far we assumed that:
\begin{align*}
	p_i(c') \leq u_i(c') \cdot \rho(c) =  u_i(c') \cdot \frac{\cost(c)}{\sum_{v_j \in S} u_j(c)} \text{.}
\end{align*}
Here, we additionally require that: 
\begin{align*}
	p_i(c') \leq u_i(c') \cdot \frac{\cost(c')}{\sum_{v_j \in S} u_j(c')} \text{.}
\end{align*}
Thus, we can use the following bound on the total spending of a given voter, $m_i$:
\begin{align*}
	m_i \leq \sum_{c' \in W} u_i(c') \cdot \min \left( \frac{\cost(c')}{\sum_{j \in S} u_j(c')},  \frac{\cost(c)}{\sum_{v_j \in S} u_j(c)}\right) \text{.}
\end{align*}
This leads us to the following definition of strong EJR+.
		
\begin{definition}[Strong EJR+ for Additive Utilities]\label{def:strong-ejr-plus}
	An outcome $W \subseteq C$ of an election $E = (C, V, b)$ satisfies strong EJR+ if for each unelected candidate $c$ and each group of voters $S \subseteq V$ with $u_S(c) > 0$ there exists a voter $v_i \in S$ such that:
	\begin{align}\label{eq:ejr}
		\sum_{\substack{c' \in W \\ u_S(c') > 0}} u_i(c') \cdot \min \left( \frac{\cost(c')}{\sum_{v_j \in S} u_j(c')},  \frac{\cost(c)}{\sum_{j \in S} u_j(c)}\right)  \geq \frac{b}{n} \text{.}
	\end{align}
	An outcome satisfies strong EJR+ up-to-one if for each $c \notin W$ and each $S \subseteq V$ with $u_S(c) > 0$ there exists $v_i \in S$ such that:
	\begin{align}\label{eq:ejr-up-to-one}
		\sum_{\substack{c' \in W\cup \{ c \} \\ u_S(c') > 0}} u_i(c') \cdot \min \left( \frac{\cost(c')}{\sum_{v_j \in S} u_j(c')},  \frac{\cost(c)}{\sum_{j \in S} u_j(c)}\right)  > \frac{b}{n} \text{.}
	\end{align}
\hfill $\lrcorner$
\end{definition}
	
\noindent
For approval utilities and cost utilities, Inequality \eqref{eq:ejr} simplifies to somewhat similar conditions:
\begin{align*}
	\underbrace{\sum_{c' \in W \cap A_i} \min\left(\frac{ \cost(c')}{|S\cap V_{c'}|}, \frac{\cost(c)}{|S\cap V_{c}|}\right) \geq \frac{b}{n}}_{\text{approval utilities}}
	\qquad \underbrace{\sum_{c' \in W \cap A_i} \min\left(\frac{ \cost(c')}{|S\cap V_{c'}|}, \frac{\cost(c')}{|S\cap V_{c}|}\right) \geq \frac{b}{n}}_{\text{cost utilities}}\text{.}
\end{align*}

We will now establish the relation between the axiom of EJR+ and EJR for additive utilities. We start by recalling the definition of the EJR axiom for additive utilities~\cite{pet-pie-sko:c:participatory-budgeting-cardinal}. 

\begin{definition}[Additive EJR and EJR up to one]\label{def:ejr}
We say that a nonempty set of voters $S \subseteq V$ is \emph{$(\alpha,T)$-cohesive} for $T\subseteq C$ and $\alpha \colon T \to \reals_{\ge 0}$, if $\cost(T)\le b \cdot \nicefrac{|S|}{n}$ and $u_i(c) \ge \alpha(c)$ for all $v_i \in S$ and $c \in T$.

A feasible outcome $W$ satisfies \emph{EJR} if for every $(\alpha, T)$-cohesive group of voters $S$, there exists $v_i \in S$ such that $u_i(W)\ge \sum_{c \in T} \alpha(c)$.
We say that $W$ satisfies \emph{EJR up to one} if for every such group of voters~$S$ there exists $v_i \in S$ such that either $u_i(W)\ge \sum_{c \in T} \alpha(c)$, or $u_i(W) + u_i(c') > \sum_{c \in T} \alpha(c)$ for some $c' \in T\setminus W$. \hfill $\lrcorner$
\end{definition}

\begin{theorem}
    Strong EJR+ (up-to-one) implies EJR (up-to-one) for additive utilities.
\end{theorem}

\begin{proof}
Let $W$ be a feasible outcome, and let us fix an $(\alpha,T)$-cohesive group of voters $S$. Without loss of generality we may assume that $\alpha(d) > 0$ for every $d \in T$. Indeed, removing candidates $d$ with $\alpha(d)=0$ preserves cohesiveness and does not change the target utility $\sum_{d \in T}\alpha(d)$.

If $T \subseteq W$, then every voter $v_i \in S$ satisfies
\begin{align*}
	u_i(W) \geq \sum_{d \in T}u_i(d) \geq \sum_{d \in T}\alpha(d),
\end{align*}
so the conclusions for EJR and EJR up-to-one follow. Hence, we can assume that $T \setminus W \neq \emptyset$.
Let $c^\star$ be the candidate from $T \setminus W$ with the minimal ratio of $\nicefrac{\cost(c^*)}{u_S(c^*)}$. Let 
\begin{align*}
        \gamma := \frac{\cost(c^\star)}{u_S(c^\star)} = \min_{d \in T \setminus W} \frac{\cost(d)}{u_S(d)}.
\end{align*}

For every voter $v_i \in S$, we have
\begin{align}
\begin{split}
&\sum_{d \in W} u_i(d) \cdot \min\left(\frac{\cost(d)}{u_S(d)},\gamma \right) \leq \sum_{d \in T \cap W} u_i(d) \cdot \min\left(\frac{\cost(d)}{u_S(d)},\gamma \right) + \sum_{d \in W \setminus T} u_i(d) \cdot \gamma \\
&\qquad\qquad\leq \sum_{d \in T \cap W} \alpha(d) \cdot \min\left(\frac{\cost(d)}{u_S(d)},\gamma \right) + \gamma \cdot \left(u_i(W) -\sum_{d \in T \cap W}\alpha(d) \right)\\ 
&\qquad\qquad = \sum_{d \in T \cap W} \alpha(d) \cdot \min\left(\frac{\cost(d)}{u_S(d)},\gamma \right)+ \gamma \cdot \sum_{d \in T \setminus W}\alpha(d)+ \gamma \cdot \left(u_i(W)-\sum_{d \in T}\alpha(d) \right)\\
&\qquad\qquad \leq \sum_{d \in T \cap W} \alpha(d) \cdot \frac{\cost(d)}{u_S(d)} + \sum_{d \in T \setminus W}\alpha(d) \cdot \frac{\cost(d)}{u_S(d)}  + \gamma \cdot \left(u_i(W)-\sum_{d \in T}\alpha(d) \right)\\
 &\qquad\qquad = \sum_{d \in T} \alpha(d)\cdot \frac{\cost(d)}{u_S(d)} + \gamma \cdot \left(u_i(W) - \sum_{d \in T}\alpha(d) \right) \leq \sum_{d \in T} \alpha(d)\cdot \frac{\cost(d)}{|S|\alpha(d)} + \gamma \cdot \left(u_i(W) - \sum_{d \in T}\alpha(d) \right)\\
 &\qquad\qquad \leq \frac{\cost(T)}{|S|}+ \gamma\left(u_i(W) - \sum_{d \in T}\alpha(d) \right) \leq \frac{b}{n} + \gamma\left( u_i(W) - \sum_{d \in T}\alpha(d) \right).
\end{split}
\label{eq:strong-ejr-utility-bound}
\end{align}

We will now prove the two implications separately, starting from the exact case, and then moving to the up-to-one case.
    
\emph{Exact case.} Suppose that $W$ satisfies strong EJR+. Applying the axiom to $S$ and $c^\star$, we obtain that there is a voter $v_i \in S$ such that
\begin{align*}
   \sum_{d \in W} u_i(d) \min \left(\frac{\cost(d)}{u_S(d)},\gamma \right) \geq \frac{b}{n} \text{.}
\end{align*}
Combining this inequality with \eqref{eq:strong-ejr-utility-bound} yields
\begin{align*}
  \frac{b}{n} \leq \frac{b}{n} + \gamma\left(u_i(W) - \sum_{d \in T}\alpha(d) \right) \text{.}
\end{align*}
Since $\gamma > 0$, it follows that
\begin{align*}
        u_i(W)\geq \sum_{d \in T}\alpha(d).
\end{align*}
Thus, $W$ satisfies EJR.

 \emph{Up-to-one case.} Suppose that $W$ satisfies strong EJR+ up to one. Applying this axiom to $S$ and $c^\star$, and combining it with \eqref{eq:strong-ejr-utility-bound} we obtain a voter $v_i \in S$ for whom
 \begin{align*}
    \frac{b}{n} &< \sum_{d \in W} u_i(d) \cdot \min\left( \frac{\cost(d)}{u_S(d)},\gamma \right) + \gamma u_i(c^\star) \\
        & \leq \frac{b}{n} + \gamma \cdot \left(u_i(W) + u_i(c^\star) - \sum_{d \in T}\alpha(d) \right) \text{,}
\end{align*}
Since $\gamma>0$, we conclude that
 \begin{align*}
	u_i(W) + u_i(c^\star) > \sum_{d \in T}\alpha(d).
 \end{align*}
As $c^\star \in T \setminus W$, this establishes EJR up to one.
\end{proof}

The demanding requirements of strong EJR+ up-to-one raise a natural question: is the axiom always satisfiable? We answer this question affirmatively by showing that the Method of Equal Shares satisfies it. In contrast, our exact formulation of EJR+ is not always satisfiable, even for approval committee elections. This justifies our future focus on the "up-to-one" variant.
		
\begin{theorem}\label{thm:mes-strong-ejr-plus}
    The Method of Equal Shares satisfies strong EJR+ up-to-one.
\end{theorem}

\begin{proof}
    Let $W$ be an outcome returned by Equal Shares, and towards a contradiction suppose that $W$ violates strong EJR+ up-to-one. Let $S\subseteq V$ and $c \in C\setminus W$ witness this violation, with $u_S(c)>0$, and let
   \begin{align*}
        \gamma = \frac{\cost(c)}{u_S(c)} > 0.
   \end{align*}
    Thus, for every voter $v_i\in S$, we have
    \begin{align}\label{eq:mes-strong-ejr-violation}
        \sum_{\substack{d\in W\\u_S(d)>0}} u_i(d) \cdot \min\left(\frac{\cost(d)}{u_S(d)},\gamma\right) + \gamma u_i(c) \leq \frac{b}{n}.
    \end{align}

    We now show, by induction on the iterations of MES, that each voter $v_i \in S$ pays for every selected candidate $d$ with $u_S(d) > 0$ at most
    \begin{align}\label{eq:mes-strong-ejr-payment}
        u_i(d) \cdot \min\left(\frac{\cost(d)}{u_S(d)}, \gamma\right).
    \end{align}
    The induction hypothesis is trivial before the first iteration. Consider an iteration in which MES selects a candidate $d$ at parameter $\rho$. Let $W'$ denote the set of candidates selected before this iteration, and let $b_i$ denote the remaining budget of voter $v_i$ immediately before selecting $d$. Without loss of generality we can assume that $u_S(d) > 0$. By the induction hypothesis, for every voter $v_i\in S$ we have that
    \begin{align*}
        b_i \geq \frac{b}{n}  - \sum_{\substack{x\in W'\\u_S(x)>0}} u_i(x) \cdot \min\left( \frac{\cost(x)}{u_S(x)}, \gamma\right) \geq
            u_i(d) \cdot \min\left(\frac{\cost(d)}{u_S(d)},\gamma\right) + \gamma u_i(c).
    \end{align*}
    The second inequality follows from \eqref{eq:mes-strong-ejr-violation}, since $d\in W\setminus W'$ and all the omitted summands are nonnegative. In particular, $b_i\geq\gamma u_i(c)$ for every $v_i\in S$. Hence,
    \begin{align*}
        \sum_{v_i\in V}\min(b_i,\gamma u_i(c)) \geq \sum_{v_i\in S}\gamma u_i(c) = \cost(c).
	\end{align*}
    Since $c \notin W$, candidate $c$ is still unelected. MES therefore selects a candidate with a parameter $\rho$ satisfying $\rho\leq\gamma$. If $\cost(d) / u_S(d) < \gamma$, the remaining-budget bound also gives
    \begin{align*}
        b_i \geq u_i(d) \cdot \frac{\cost(d)}{u_S(d)} \qquad\text{for every }v_i\in S.
	\end{align*}
    Consequently,
    \begin{align*}
        \sum_{v_i \in V} \min\left(b_i, \frac{\cost(d)}{u_S(d)} \cdot u_i(d)\right) \geq \sum_{v_i\in S} \frac{\cost(d)}{u_S(d)} \cdot u_i(d) = \cost(d).
	\end{align*}
    Thus, $d$ itself is affordable at a parameter at most $\cost(d)/u_S(d)$. In either case, we have that
    \begin{align*}
        \rho \leq \min\left(\frac{\cost(d)}{u_S(d)}, \gamma\right).
    \end{align*}
    Each voter $v_i\in S$ therefore pays
    \begin{align*}
        \min(b_i,\rho u_i(d)) \leq u_i(d) \cdot \min\left(\frac{\cost(d)}{u_S(d)},\gamma\right),
	\end{align*}
    which completes the induction.

    Finally, after all candidates in $W$ have been selected, the payment bounds and \eqref{eq:mes-strong-ejr-violation} imply that every voter $v_i\in S$ has remaining budget
    \begin{align*}
        b_i \geq \frac{b}{n} - \sum_{\substack{d\in W\\u_S(d)>0}} u_i(d) \cdot \min\left(\frac{\cost(d)}{u_S(d)}, \gamma\right) \geq \gamma u_i(c) \text{.}
    \end{align*}
    Hence, $c$ is still affordable, which contradicts the stopping condition of MES and completes the proof.
\end{proof}

\begin{example}[Exact strong EJR+ is not always satisfiable]
    Consider a committee election with two voters, two unit-cost candidates $C=\{c_1,c_2\}$, and budget $b=1$. The voters have approval utilities, with
    \begin{align*}
        A_1=\{c_1\},
        \qquad
        A_2=\{c_2\}.
	\end{align*}
    Let $W$ be a feasible outcome. Since $|W|\leq 1$, there exists $i\in\{1, 2\}$ such that $c_i \notin W$. Consider the group $S = \{v_i\}$ and the unelected candidate $c_i$. We have $u_S(c_i)=1$, but $u_S(d)=0$ for every $d\in W$. Consequently,
    \begin{align*}
        \sum_{\substack{d\in W\\u_S(d)>0}} u_i(d) \cdot \min\left(\frac{\cost(d)}{u_S(d)}, \frac{\cost(c_i)}{u_S(c_i)} \right) = 0 < \frac{1}{2} = \frac{b}{n}.
	\end{align*}
    Thus, the axiom is not always satisfiable, even for committee elections with approval utilities. \hfill $\lrcorner$
\end{example}

\subsection{Strong EJR+ versus Weak EJR+}

It is worth noting that even for cost-utilities and for committee elections with approval utilities the definition of strong EJR+ is considerably stronger than the definition of EJR+ as proposed by~\citet{BP-ejrplus}. This is illustrated by the following example.

\begin{example}[Shared winners can make strong EJR+ more demanding]
\label{ex:shared}
There are ten candidates in each of three clusters $A$, $B$, and $C$, all of unit cost, and the budget equals to $20$. There are $1000$ voters, and their approval sets are given in \cref{tab:shared}. Take $W=A \cup C$.

\begin{table}[htbp]
\centering
\caption{The approval profile from \cref{ex:shared}. Each group, $A$, $B$, and $C$, contains ten candidates.}
\label{tab:shared}
\begin{tabular}{lrrr}
\toprule
Voter group & Number of voters & Approval set & Approval utility from $W$\\
\midrule
First group & $500$ & $A\cup B$ & $10$\\
Second group & $499$ & $A$ & $10$\\
Last voter & $1$ & $A\cup B\cup C$ & $20$\\
\bottomrule
\end{tabular}
\end{table}

Observe that $W$ satisfies standard EJR+. Only $B$-candidates are not selected. A coalition supporting one of them either contains the last voter, who has utility $20$, or consists of at most $500$ voters each having utility $10$. These utilities meet the requirement from \Cref{def:standard-ejr-plus}.

The outcome is also in the core. A coalition containing a second-group voter cannot deviate since these voters have already attained their maximum level of satisfaction. A coalition
containing the last voter cannot deviate, since this voter already approves 20 selected candidates from the solution $W$, which is the maximal number of candidates that can be selected with the budget of $20$. Every remaining coalition consists of at most $500$ first-group voters and can afford at most ten candidates, which again precludes a strict improvement.

Nevertheless, let $S$ consist of the first $999$ voters and choose $c\in B$. Then 
\begin{align}
	\sum_{c' \in W\cup \{ c \} \colon u_S(c') > 0} u_i(c') \cdot \min \left( \frac{\cost(c')}{\sum_{v_j \in S} u_j(c')}, \frac{\cost(c)}{\sum_{j \in S} u_j(c)}\right) \leq \frac{10}{999} + \frac1{500} < \frac{20}{1000} \text{.}
\end{align}

Thus $W$ fails strong EJR+ up to one. This example also explains why non-supporters of the challenger candidate $c$ must be allowed in $S$: the second group helps finance the unanimously supported $A$-candidates, and so the first one can claim the right to some additional candidates from $B$. \hfill $\lrcorner$
\end{example}
		
Intuitively, in \Cref{ex:shared}, we see that the unanimously approved $A$-candidates occupy half of the slots in the committee. The first group of half of the voters is still entitled to roughly half of the remaining slots, which is well-aligned with our intuitive interpretation of fairness. Notably, the same example shows that strong EJR+ up-to-one is not implied by the core, and it addresses the limitations of the core mentioned by~\citet{pet-pie-sha-sko:stable-priceability}.
	
\subsection{Verification of the EJR+ Axioms}
\label{sec:verification-ejr-plus}

Although the definitions of weak and strong EJR+ quantify over exponentially many groups of voters, all four variants can be verified in polynomial time. Our algorithm considers each unelected candidate separately and repeatedly removes voters who cannot belong to any group witnessing a violation. The key observation is that shrinking a group can only increase the payment bounds of the voters who remain in that group.

\begin{theorem}\label{thm:verification-ejr-plus}
    Given an election and a feasible outcome $W\subseteq C$, it can be verified in $O\bigl(mn^2(|W|+1)\bigr)$ time whether $W$ satisfies weak EJR+, weak EJR+ up-to-one, strong EJR+, or strong EJR+ up-to-one.
\end{theorem}
\begin{proof}
    We describe the algorithm for strong EJR+ up-to-one, and it can be adapted in a straightforward way to the other three variants.

	Let us start with the following observation. Consider an unelected candidate $c\in C\setminus W$, two groups of voters, $S'\subseteq S\subseteq V$, with $u_{S'}(c) > 0$, and a voter $v_i \in S'$. Observe that:
    \begin{align}
        \sum_{\substack{d\in W\cup\{c\} \\ u_{S'}(d) > 0}} u_i(d) \cdot \min\left(\frac{\cost(d)}{u_{S'}(d)}, \frac{\cost(c)}{u_{S'}(c)} \right) \geq
        \sum_{\substack{d\in W\cup\{c\} \\ u_S(d) > 0}} u_i(d) \cdot \min\left(\frac{\cost(d)}{u_S(d)}, \frac{\cost(c)}{u_S(c)} \right) \text{.}
        \label{eq:verification-monotonicity}
    \end{align}
	
	Now, the algorithm works as follows. For each candidate $c\in C\setminus W$, independently, we initialize $S = V$. While $u_S(c) > 0$, we look for a voter $v_i\in S$
    satisfying
    \begin{align}
        \sum_{\substack{d\in W\cup\{c\} \\ u_S(d) > 0}} u_i(d) \cdot \min\left(\frac{\cost(d)}{u_S(d)}, \frac{\cost(c)}{u_S(c)}\right) > \frac{b}{n}.
        \label{eq:verification-removal}
    \end{align}
    If such a voter exists, remove her from $S$ and repeat. The expressions are recomputed for the new group, and we continue in the loop.

    If no voter satisfies \eqref{eq:verification-removal} and $u_S(c) > 0$, we claim $(S,c)$ as a witness that $W$ violates strong EJR+ up-to-one. If instead $u_S(c) = 0$, we proceed to the next unelected candidate. If all unelected candidates have been considered without finding a violation, then we declare that $W$ satisfies strong EJR+ up-to-one.

    Let us now prove the correctness of the algorithm. Fix a candidate $c\in C\setminus W$. We show that throughout the procedure, every group witnessing a violation involving $c$ is contained in the current set $S$. This holds initially because $S = V$. Suppose that the invariant holds and the algorithm removes a voter $v_i$ satisfying \eqref{eq:verification-removal}. Let $X\subseteq S$ be any group witnessing a violation involving $c$; by definition $u_X(c) > 0$. If $v_i\in X$, \eqref{eq:verification-monotonicity}, applied with $S'=X$, implies
   \begin{align*}
        \sum_{\substack{d\in W\cup\{c\} \\ u_X(d) > 0}} u_i(d) \cdot \min\left(\frac{\cost(d)}{u_X(d)}, \frac{\cost(c)}{u_X(c)}\right) > \frac{b}{n}.
   \end{align*}
    This contradicts the assumption that $X$ witnesses a violation. Thus, $v_i\notin T$, and the invariant is preserved.

    If the procedure ends with $u_S(c)=0$, no subset of $S$ has positive utility for $c$. By the invariant, no group witnessing a violation involving $c$ exists. Otherwise, the procedure ends with $u_S(c) > 0$ and no voter satisfying \eqref{eq:verification-removal}. Hence, every voter in $S$ has the expression in that inequality at most $\nicefrac{b}{n}$ so $(S, c)$ indeed witnesses a violation. This proves the correctness.

    Finally, let us assess the running time of the algorithm. For a fixed candidate $c$, at most $n$ voters are removed. In each iteration, all values $u_S(d)$ for $d\in W\cup\{c\}$ can be computed in $O\bigl(n(|W|+1)\bigr)$ time. Evaluating the removal condition for all remaining voters takes the same asymptotic number of operations. Thus, processing one candidate can be done in $O\bigl(n^2(|W|+1)\bigr)$ time. There are at most $m$ unelected candidates, giving the claimed bound.
\end{proof}

Note that the algorithm in \Cref{thm:verification-ejr-plus} will run considerably faster in practice if in each step we remove all voters satisfying \eqref{eq:verification-removal} from the group $S$.

\section{Stable Priceability and Core+}
\label{sec:stable-priceability-core-plus}

The EJR+ axioms considered in the previous section compare hypothetical payments for individual groups of voters; these hypothetical payments might be different for different candidates $c \in C \setminus W$ that we consider. We now consider a complementary approach, in which a single payment system provides a certificate for the whole electorate. We first describe an extension of stable priceability~\cite{pet-pie-sha-sko:stable-priceability} to additive utilities, and then we use a similar approach to extend the recently defined axiom of Core+~\cite{bec-gre-pet:core}.

Let us first recall the definition of the core. 

\begin{definition}
A feasible outcome $W$ is in the \myemph{core} if for every nonempty group of voters $S\subseteq V$ and every set of candidates $T\subseteq C$ with $\cost(T)\leq b\cdot \nicefrac{|S|}{n}$, there exists a voter $v_i\in S$ such that $u_i(W) \geq u_i(T)$.
We say that $W$ satisfies the \myemph{core up-to-one} if for every such $S$ and $T$, there exists $v_i\in S$ such that either $u_i(W)\geq u_i(T)$, or $u_i(W) + u_i(c) > u_i(T)$ for some candidate $c\in T\setminus W$.
\end{definition}

For approval committee elections, the definitions of core and core up-to-one are equivalent: any strictly improving deviation gives each member at least one additional approved candidate.

\subsection{Stable Priceability for Additive Utilities}

Stable priceability was introduced by \citet{pet-pie-sha-sko:stable-priceability} for approval committee elections. Below we extend it to additive utilities, once again following the intuition that the purchases with lower price-per-utility ratios are considered generally better.

\begin{definition}[Stable Priceability for Additive Utilities]\label{def:rate-stable-priceability}
    A feasible outcome $W$ is priceable if there exist a common virtual endowment $\beta\geq b/n$, nonnegative payments $p_i(d)$ for all $v_i\in V$ and $d\in W$, and a set of nonnegative remaining budgets of the voters $r_i$, $v_i \in V$, satisfying the following conditions
    \begin{align*}
        & r_i + \sum_{d\in W}p_i(d) = \beta &&\text{for every } v_i\in V,\\
        & p_i(d) = 0 &&\text{whenever } u_i(d) = 0,\\
        & \sum_{v_i\in V}p_i(d) =\cost(d) &&\text{for every } d\in W,
    \end{align*}
    and such that, for every $c\in C\setminus W$,
    \begin{align}\label{eq:rate-stability}
        \sum_{v_i\in V_c} \max\left(r_i, u_i(c)\cdot \max_{d\in W\cap A_i}\frac{p_i(d)}{u_i(d)}\right) \leq \cost(c) \text{.}
    \end{align}
    Here and below, a maximum over an empty set is interpreted as zero. \hfill $\lrcorner$
\end{definition}

The second expression inside the maximum in \eqref{eq:rate-stability} is the price of $u_i(c)$ units of utility, assuming it is paid at the largest rate that voter $v_i$ pays for a selected candidate. Note that the value of the maximum can be in principle higher than the voter's total budget $\beta$. Thus, it would be tempting to rewrite \eqref{eq:rate-stability} as:
\begin{align*}
    \sum_{v_i\in V_c} \max\left(r_i, \min\left(\beta, u_i(c)\cdot \max_{d\in W\cap A_i}\frac{p_i(d)}{u_i(d)}\right)\right) \leq \cost(c) \text{.}
\end{align*}
While such formulation would make the axiom a bit weaker, all the logical implications that we further show would still hold. However such formulation would result in losing polynomial-time verification, which is the reason why we use \eqref{eq:rate-stability} in \Cref{def:rate-stable-priceability}.

For approval utilities and unit costs, condition \eqref{eq:rate-stability} reduces to
\begin{align*}
    \sum_{v_i\in V_c} \max\left(r_i, \max_{d\in W\cap A_i}p_i(d)\right) \leq 1 \text{,}
\end{align*}
which is exactly the condition in the original definition of stable priceability~\cite{pet-pie-sha-sko:stable-priceability}.

Note that $\beta$ is a virtual endowment of the voters, whereas $b$ is the total available budget determining feasibility. In particular, $\beta$ may exceed $\nicefrac{b}{n}$ while the selected outcome does not exceed the available budget.

\subsection{Core+ for Additive Utilities}

Let us start by recalling the definition of loose payment systems and the definition of approval Core+, recently given by \citet[Theorem~3.3]{bec-gre-pet:core}. In contrast to the classic notion of priceability, the payment certificate in the definition of Core+ does not require the selected candidates to be paid for in full. We retain this feature in the additive setting.

\begin{definition}[Loose Payment System]
\label{def:core-plus-loose-payments}
    A \myemph{loose payment system} for a feasible outcome $W$ consists of nonnegative payments $p_i(d)$ for $v_i\in V$ and $d\in W$, and nonnegative remaining amounts $r_i$, satisfying the following conditions:
    \begin{align*}
        & r_i + \sum_{d\in W}p_i(d) = \frac{b}{n} && \text{for every } v_i \in V, \\
        & p_i(d) = 0 &&\text{whenever } u_i(d) = 0, \\
        &\sum_{v_i \in V}p_i(d) \leq \cost(d) &&\text{for every } d \in W.
    \end{align*}
\hfill $\lrcorner$
\end{definition}

Consider a loose payment system, a candidate $c \notin W$, and a voter $v_i\in V_c$. We define $\rho_i(c)$ as the solution to the following equality: 
\begin{align}\label{eq:core-plus-rebalancing}
    \sum_{d\in W}\min(p_i(d),\rho_i(c)u_i(d)) +\rho_i(c)u_i(c) = \frac{b}{n}.
\end{align}
This equation uniquely defines $\rho_i(c) > 0$. Indeed, its left-hand side, viewed as a function of $\rho_i(c)$, is continuous and strictly increasing, equals zero at zero, and tends to infinity.

The intuitive interpretation of this value is the following. Imagine that the voter reduces the payments for selected candidates whenever their payment-per-utility ratios exceed $\rho_i(c)$. She then can use her remaining and the released money to pay for $c$ with the same payment-per-utility $\rho_i(c)$. Indeed, it can be equivalently written as:
\begin{align*}
    \rho_i(c)u_i(c) = r_i + \sum_{d\in W}\max(p_i(d)-\rho_i(c)u_i(d),0) \text{.}
\end{align*}
Thus, for each loose payment system we can uniquely define the values $\rho_i(c)$ for $v_i \in V$ and $c \notin W$. This leads us to the following definition of Core+.

\begin{definition}[Core+ up-to-one for Additive Utilities]\label{def:additive-core-plus}
    A feasible outcome $W$ satisfies \myemph{Core+} if it admits a loose payment system such that, for every candidate $c\in C\setminus W$ it holds that:
    \begin{align}\label{eq:core-plus-losing-candidate}
        \sum_{v_i \in V_c} \rho_i(c)u_i(c) < \cost(c),
    \end{align}
    where $\rho_i(c)$ is defined as a solution to \eqref{eq:core-plus-rebalancing}.
\end{definition}

Note that in \Cref{def:additive-core-plus} the payments for selected candidates are the same for all challengers $c \notin W$.

\begin{remark}
    For approval committee elections and outcomes of size exactly $b$, \Cref{def:additive-core-plus} is equivalent to Core+ as defined by \citet{bec-gre-pet:core}.
\end{remark}

\begin{theorem}\label{thm:core-plus-implies-core-up-to-one}
    Core+ up-to-one implies the core up-to-one for additive utilities.
\end{theorem}
\begin{proof}
    Let $W$ be an outcome that satisfies Core+, and fix a witnessing loose payment system. Suppose, towards a contradiction, that there is a nonempty group $S\subseteq V$ and a set of candidates $T\subseteq C$ such that $\cost(T)\leq b \cdot \nicefrac{|S|}{n}$, $u_i(T) > u_i(W)$  for every $v_i\in S$, and $u_i(W) + u_i(c) \leq u_i(T)$ for all $c\in T\setminus W$.

    For each voter $v_i\in S$ we define $c_i$ as the candidate from $T\setminus W$ that maximizes $u_i(c_i)$. In particular, $u_i(c_i) > 0$. For every $c \in (T\setminus W) \cap A_i$, we have that $\rho_i(c) \geq \rho_i(c_i)$.  Indeed, if we replace $u_i(c_i)$ by the weakly smaller value $u_i(c)$ in \eqref{eq:core-plus-rebalancing}, then the slope of the left-hand side function can only decrease.

    Now, for every $v_i\in S$, we obtain
    \begin{align*}
        \frac{b}{n} &=\sum_{d\in W} \min(p_i(d),\rho_i(c_i)u_i(d)) +\rho_i(c_i)u_i(c_i) \leq \sum_{d\in T\cap W}p_i(d) + \sum_{d\in W \setminus T}\rho_i(c_i)u_i(d) +\rho_i(c_i)u_i(c_i) \\
        & = \sum_{d\in T\cap W}p_i(d) + \rho_i(c_i)\bigl(u_i(W\setminus T) + u_i(c_i)\bigr) \leq\sum_{d\in T\cap W}p_i(d) + \rho_i(c_i)u_i(T\setminus W)\\
        & \leq\sum_{d\in T\cap W}p_i(d)
            +\sum_{c\in (T\setminus W)\cap A_i}
                \rho_i(c)u_i(c).
    \end{align*}
    Summing over $S$, we get
    \begin{align*}
        |S|\frac{b}{n} &\leq \sum_{d\in T\cap W}\sum_{v_i \in S}p_i(d) +\sum_{c\in T\setminus W} \sum_{v_i \in S\cap V_c}\rho_i(c) u_i(c) < \cost(T\cap W) + \cost(T\setminus W) = \cost(T) \leq |S| \cdot \frac{b}{n} \text{.}
    \end{align*}
    The strict inequality follows from \eqref{eq:core-plus-losing-candidate} and the fact that $T\setminus W\neq\emptyset$. This gives a contradiction, and completes the proof.
\end{proof}

Below we show that stable-priceability is a stronger notion than Core+ up-to-one. 

\begin{theorem}\label{thm:rate-stability-implies-core-plus}
    Suppose that $W$ admits a stable price system with common virtual endowment $\beta \geq \nicefrac{b}{n}$. If either $\beta > \nicefrac{b}{n}$ or $W$ is exhaustive, then $W$ satisfies Core+ up-to-one.
\end{theorem}
\begin{proof}
    Fix a stable price system. For $c \notin W$ and $v_i\in V_c$ we set
    \begin{align*}
        M_i(c) = \max\left(r_i, u_i(c) \cdot \max_{d \in W\cap A_i} \frac{p_i(d)}{u_i(d)} \right) \text{.}
    \end{align*}
    Hence, by the condition of stable prcieability it holds that $\sum_{v_i\in V_c} M_i(c) \leq \cost(c)$. Let $\lambda = \nicefrac{b}{(n\beta)}\leq 1$, and scale the payments and remaining amounts by setting $p'_i(d)=\lambda p_i(d)$ and
    $r'_i=\lambda r_i$. They form a loose payment system with endowment $\nicefrac{b}{n}$.

    Let $\rho'_i(c)$ be computed through \eqref{eq:core-plus-rebalancing} for this new system. By the definition of $M_i(c)$ we have
    \begin{align*}
        \sum_{d\in W} \min\left(p'_i(d), \frac{\lambda M_i(c)}{u_i(c)}u_i(d)\right) + \lambda M_i(c) =\lambda\sum_{d\in W}p_i(d) + \lambda M_i(c) \geq \frac{b}{n} \text{.}
    \end{align*}
    From \eqref{eq:core-plus-rebalancing} we can therefore infer that $\rho'_i(c)u_i(c)\leq\lambda M_i(c)$ (this follows from the fact that the left-hand side of \eqref{eq:core-plus-rebalancing} is strictly increasing as the function of $\rho$). If $\beta>b/n$, then $\lambda<1$, and
    \begin{align*}
        \sum_{v_i \in V_c}\rho'_i(c)u_i(c) \leq \lambda\sum_{v_i \in V_c} M_i(c) \leq \lambda \cost(c) < \cost(c) \text{,}
    \end{align*}
    which is the condition required by Core+ up-to-one.

    It remains to consider the case when $\beta = \nicefrac{b}{n}$ and when $W$ is exhaustive. Here $\lambda = 1$. If $M_i(c) > r_i$, then the previous estimation also gives a strict inequality. Consequently, equality $\sum_{v_i\in V_c}\rho_i(c)u_i(c)=\cost(c)$ could hold only if $\rho_i(c)u_i(c) = M_i(c) = r_i$ for every $v_i\in V_c$.
    But this would give
    \begin{align*}
        \cost(c) = \sum_{v_i\in V_c}r_i \leq\sum_{v_i\in V}r_i = b - \cost(W) \text{,}
    \end{align*}
    contradicting exhaustiveness. Thus, the required strict inequality holds for every unelected candidate, which completes the proof.
\end{proof}

\subsection{Verification of Stable Priceability and Core+}

For a fixed outcome, the definition of stable priceability can be expressed as a linear program. Indeed, we can introduce variables $M_i(c)$ for $c\notin W$ and $v_i\in V_c$, and the following constraints:
\begin{align*}
    M_i(c)\geq r_i, \qquad  M_i(c)\geq\frac{u_i(c)}{u_i(d)} \cdot p_i(d) \quad\text{for every } d\in W\cap A_i,
\end{align*}
Then, it is sufficient to replace \eqref{eq:rate-stability} with $\sum_{v_i\in V_c}M_i(c)\leq\cost(c)$. All these constraints are linear, and hence the axiom is verifiable in polynomial time~\cite{ellipsoidMethod}. This also illustrates that for approval-based committee elections our formulation of stable-priceability is equivalent to the original formulation by \citet{pet-pie-sha-sko:stable-priceability}.

Verifying Core+ up-to-one requires slightly more care, since \eqref{eq:core-plus-losing-candidate} is a strict inequality. To handle this, we introduce a common slack variable $\epsilon\in[0,1]$ and replace this condition by
\begin{align*}
    \sum_{v_i \in V_c} \rho_i(c)u_i(c) \leq (1-\epsilon) \cdot \cost(c) \qquad\text{for every } c\in C\setminus W.
\end{align*}
Additionally, the $\min$ operator in \eqref{eq:core-plus-rebalancing} can be handled by introducing additional variables and constraints. We introduce nonnegative variables $z_i(d, c)$ for $d\in W$, and impose that:
\begin{align*}
    &z_i(d, c) \leq p_i(d), \\
    &z_i(d, c) \leq \rho_i(c)u_i(d), \\
    &\sum_{d\in W}z_i(d,c) + \rho_i(c)u_i(c) \geq \frac{b}{n}.
\end{align*}

All other constraints remain unchanged. Since all costs are strictly positive, a payment system satisfies the original strict inequalities if and only if it satisfies these non-strict inequalities for some common $\epsilon > 0$. We therefore maximize $\epsilon$ subject to the new constraints. The outcome satisfies Core+ up-to-one exactly when the resulting linear program is feasible and has a strictly positive optimum.

\begin{theorem}\label{thm:verify-additive-core-plus}
    It can be verified in polynomial time whether a feasible outcome is stable priceable.  It can be verified in polynomial time whether a feasible outcome satisfies Core+ up-to-one.
\end{theorem}

\section{From EJR+ to PAV-Inspired Selection Rules} \label{sec:ejr-pav-connection}

In this section, we describe a connection between our formulation of weak EJR+ up-to-one and the objective of Proportional Approval Voting. This connection can be used to design other proportional selection rules.

\subsection{Harmonic Means and Marginal PAV Scores}

Consider a feasible outcome $W$. Recall that, in the definition of weak EJR+ up-to-one, it suffices to consider nonempty groups $S\subseteq V_c$. The axiom requires that for every $c \in C\setminus W$ and every nonempty $S\subseteq V_c$ the following condition should hold:
\begin{align}\label{eq:pav-weak-maximum}
    \max_{v_i\in S} \frac{\cost(c)}{u_S(c)} \bigl(u_i(W)+u_i(c)\bigr) > \frac{b}{n}
\end{align}

Let us replace the maximum by the harmonic mean of the same expressions. This gives the following condition:
\begin{align}\label{eq:pav-harmonic-condition}
    \frac{|S|}{\displaystyle\sum_{v_i\in S} \frac{u_S(c)}{\cost(c)\bigl(u_i(W)+u_i(c)\bigr)}} > \frac{b}{n} \text{.}
\end{align}
Since the harmonic mean of values is at most the maximum of these values, condition \eqref{eq:pav-harmonic-condition} is a strengthening of \eqref{eq:pav-weak-maximum}.

Consider now committee elections with approval utilities. Since $S\subseteq V_c$, we have $u_S(c)=|S|$, $u_i(c)=1$, and $\cost(c)=1$. Thus, the harmonic mean simplifies to
\begin{align*}
    \frac{|S|}{\displaystyle\sum_{v_i \in S} \frac{|S|}{|A_i\cap W| + 1}} = \left(\sum_{v_i\in S}\frac{1}{|A_i\cap W|+1}\right)^{-1} \text{.}
\end{align*}
The sum inside the parentheses increases when voters are added to $S$. Consequently, it suffices to check $S=V_c$, and \eqref{eq:pav-harmonic-condition} is equivalent to
\begin{align}\label{eq:pav-bounded-improvement}
    \sum_{v_i\in V_c}\frac{1}{|A_i\cap W| + 1}<\frac{n}{b} \qquad \text{for every } c\in C \setminus W.
\end{align}
The left-hand side is exactly the marginal increase of the PAV score when adding $c$ to $W$. Let us denote this increase by
\begin{align}\label{eq:pav-marginal-increase}
    \Delta(c\mid W) &:=\operatorname{sc}_{\mathrm{PAV}}(W\cup\{c\}) - \operatorname{sc}_{\mathrm{PAV}}(W) = \sum_{v_i \in V_c}\frac{1}{|A_i\cap W| + 1}.
\end{align}
Condition~\eqref{eq:pav-bounded-improvement} is known from the analysis of proportionality guarantees of the local-search algorithm for PAV~\cite{AEHLSS18,justifiedRepresentation,BP-ejrplus}. Below, we recsll the argument and make the relation between our condition and PAV explicit.

\begin{proposition}\label{prop:pav-marginal-bound}
    Let $W$ be a size-$b$ committee such that no exchange of a selected candidate for an unelected candidate increases its PAV score. Then, for every $c \in C\setminus W$ it holds that
    \begin{align*}
        \Delta(c\mid W) \leq \frac{n}{b + 1} < \frac{n}{b}.
    \end{align*}
    In particular, every PAV outcome satisfies \eqref{eq:pav-harmonic-condition}.
\end{proposition}

\begin{proof}
    Fix $c \in C\setminus W$ and let $W' = W\cup \{c\}$. The sum of the losses caused by removing individual candidates from $W'$ is
    \begin{align*}
        \sum_{d\in W'}\left(\operatorname{sc}_{\mathrm{PAV}}(W') - \operatorname{sc}_{\mathrm{PAV}}(W' \setminus \{d\})\right) = \sum_{\substack{v_i\in V\\A_i\cap W'\neq\emptyset}} \sum_{d\in A_i\cap W'}\frac{1}{|A_i\cap W'|} = |\{v_i\in V \colon A_i \cap W' \neq \emptyset\}| \leq n \text{.}
    \end{align*}
    Removing $c$ decreases the score by exactly $\Delta(c \mid W)$. For every $d \in W'$, local optimality gives
    \begin{align*}
        \operatorname{sc}_{\mathrm{PAV}}(W') - \operatorname{sc}_{\mathrm{PAV}}(W' \setminus \{d\}) \geq \Delta(c\mid W).
    \end{align*}
    By summing-up over all $d \in W'$ we get that $(b+1)\Delta(c\mid W)\leq n$, which completes the proof.
\end{proof}

\subsection{Minimizing the Largest Marginal Gain}

The preceding observation suggests defining a selection rule directly from the marginal-gain bound, rather than by maximizing the total PAV score. For each size-$b$ committee $W$, let
\begin{align}\label{eq:pav-minimax-objective}
    \Phi(W)=\max_{c\in C\setminus W}\Delta(c\mid W),
\end{align}
where a maximum over an empty set is interpreted as zero. The proposed rule selects all size-$b$ committees minimizing $\Phi(W)$. Intuitively, it minimizes the largest remaining marginal gain offered by an unelected candidate.

\begin{corollary}\label{cor:pav-minimax-ejr}
    Every outcome of the proposed rule that minimizes $\Phi(W)$ satisfies EJR+ and hence EJR.
\end{corollary}

\begin{proof}
    A size-$b$ PAV outcome witnesses that the minimum of $\Phi$ is at most $\nicefrac{n}{(b+1)}$. Thus, every outcome minimizing $\Phi(W)$ satisfies \eqref{eq:pav-bounded-improvement}, and consequently the harmonic-mean strengthening of weak EJR+ up-to-one. Thus, it satisfies EJR+ for approval committee elections.
\end{proof}

This rule minimising $\Phi(W)$ is inspired by PAV, but does not always return a PAV-optimal committee. Consider six voters, three candidates $C = \{a, b, c\}$, and a committee size of one, with the following approval ballots.
\begin{align*}
    A_1=A_2=\{a\},\qquad
    A_3=A_4=\{b\},\qquad
    A_5=\{a,c\},\qquad
    A_6=\{b,c\}.
\end{align*}
PAV selects $\{a\}$ or $\{b\}$, each with score $3$, whereas $\{c\}$ has score $2$. However,
\begin{align*}
    \Phi(\{a\}) = \Phi(\{b\}) = 3, \qquad \Phi(\{c\}) = \frac{5}{2}.
\end{align*}
Thus, our minimax rule uniquely selects $\{c\}$. 

Note that the optimisation objective of $\Phi(W)$ does not naturally admit a polynomial-time algorithm.

\subsection{Extensions to Additive Utilities}

The same construction suggests an optimization objective for general additive utilities. Here, one can use a weighted harmonic mean, assigning weight $u_i(c)$ to voter $v_i\in S$. This coincides with the unweighted harmonic mean for approval utilities. The resulting mean is
\begin{align*}
    \frac{\cost(c)}{\displaystyle\sum_{v_i\in S} \frac{u_i(c)}{u_i(W)+u_i(c)}}.
\end{align*}
The expression is lowest for $S=V_c$. We again take the reciprocal of the value above, and obtain the following objective:
\begin{align}\label{eq:additive-weighted-harmonic-objective}
    \Phi_{\mathrm{wH}}(W) = \max_{c\in C\setminus W} \frac{1}{\cost(c)} \sum_{v_i\in V_c} \frac{u_i(c)}{u_i(W)+u_i(c)} \text{.}
\end{align}
This objective also reduces to $\Phi(W)$ for approval committee elections. In addition, it is invariant under independently rescaling all utilities of any voter by a positive constant. Any outcome with $\Phi_{\mathrm{wH}}(W)< \nicefrac{n}{b}$ satisfies weak EJR+ up-to-one, since every weighted harmonic mean is at most the corresponding maximum.

There is also a direct connection with Core+ up-to-one. For a loose payment system witnessing that axiom,\eqref{eq:core-plus-rebalancing} gives
\begin{align*}
    \frac{b}{n} = \sum_{d\in W}\min\{p_i(d),\rho_i(c)u_i(d)\} + \rho_i(c)u_i(c) \leq\rho_i(c)\bigl(u_i(W) + u_i(c)\bigr).
\end{align*}
By combining this with \eqref{eq:core-plus-losing-candidate} we get that for every unelected candidate $c$ it holds that:
\begin{align*}
    \frac{b}{n}\sum_{v_i\in V_c} \frac{u_i(c)}{u_i(W) + u_i(c)} \leq \sum_{v_i\in V_c}\rho_i(c)u_i(c) < \cost(c) \text{.}
\end{align*}
Thus, Core+ up-to-one would imply $\Phi_{\mathrm{wH}}(W) < \nicefrac{n}{b}$.

These constructions motivate studying PAV-inspired rules defined through minimisation of certain expressions over non-elected candidates. Analogous objectives can be considered starting with the strong EJR+, but their formulation must specify how to treat nonsupporters of the challenger candidate $c \notin W$ and the exact alternative in the up-to-one axiom. Establishing existence and computational guarantees for such extensions requires a separate analysis.

\section{Conclusion}
\label{sec:conclusion}

We have proposed extensions of the axioms of EJR+ and Core+ to the model of participatory budgeting with additive utilities. Our approach builds on the idea that voters should control equal shares of the budget and that they would compare purchases according to their payment-per-utility ratios. This provides a common interpretation of proportionality across approval utilities, cost utilities, and general additive utilities. Similarly to their original counterparts, our notions formulate representation requirements that do not rely on cohesive groups of voters.

For EJR+, we have introduced weak and strong variants. The weak up-to-one variant recovers the existing definitions for approval committee elections and cost utilities. The strong up-to-one variant implies additive EJR up to one. Moreover, even for approval committee elections, it excludes some outcomes that satisfy both standard EJR+ and the core, and that can be intuitively considered non-proportional. Thus, our payment-based approach not only extends existing notions to a more general model, but also strengthens known proportionality axioms within the classic approval-based setting. All four EJR+ variants considered in the paper can be verified in polynomial time. While the exact variants of the axioms are not always satisfiable, their up-to-one counterparts are satisfied by the outcomes of the method of equal shares.  

Our formulation of Core+ up-to-one is also based on the idea of payments, yet it uses what we call loose payment systems, where the costs of the selected candidates do not need to be paid by the voters in full.  We also established a connection with our extension of stable priceability: every exhaustive stably priceable outcome satisfies Core+ up-to-one. Both axioms can be verified in polynomial time through a linear-programming formulation.

We have also explored a connection between our formulation of EJR+ and Proportional Approval Voting. For approval committee elections, replacing the maximum in weak EJR+ up-to-one by a harmonic mean yields a bound on the marginal PAV score of each unelected candidate. This condition is satisfied by PAV and motivates a different selection rule that minimizes the largest remaining marginal gain. Although this rule need not return
PAV-optimal committees, it still satisfies EJR+. The corresponding extensions to additive utilities are possible, but more delicate. This motivates further research on how payment-based proportionality requirements can
guide the design of PAV-inspired rules for additive utilities.

An important direction for future work is the empirical evaluation of the proposed axioms. Our verification algorithms make it possible to examine how often different selection rules satisfy these requirements on election data, including rules that do not provide worst-case guarantees. In particular, it would be useful to understand how frequently strong EJR+ up-to-one detects violations that are not identified by existing axioms, and how these differences depend on the structure of voters' preferences.

\subsection*{Acknowledgments}

Piotr Skowron is supported by the European Union (ERC, PRO-DEMOCRATIC, 101076570). Views and opinions expressed are however those of the author(s) only and do not necessarily reflect those of the European Union or the European Research Council. Neither the European Union nor the granting authority can be held responsible for them.
\begin{center}
  \includegraphics[width=3cm]{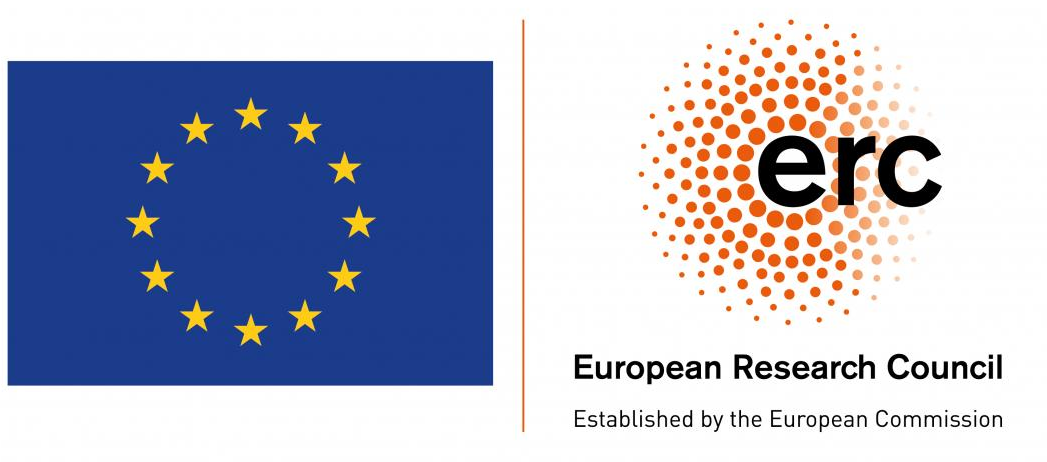}
\end{center}

\bibliographystyle{plainnat}
\bibliography{prop-additive-axioms}
		
\end{document}